\documentclass[submission,copyright]{eptcs}
\providecommand{\event}{Tenth Workshop on Model-Based Testing, April 18, 2015} 
\usepackage{breakurl}             

\newtheorem{theorem}{Theorem}[section]

\newtheorem{proposition}[theorem]{Proposition}
\newtheorem{corollary}[theorem]{Corollary}
\newenvironment{proof}[1][Proof]{\begin{trivlist}
\item[\hskip \labelsep {\bfseries #1}]}{\end{trivlist}}

\title{Adaptive Homing is in P}
\author{Natalia Kushik
\institute{Tomsk State University\\ Tomsk, Russia}
\email{ngkushik@gmail.com}
\and
Nina Yevtushenko
\institute{Tomsk State University\\ Tomsk, Russia}
\email{yevtushenko@sibmail.com}
}
\def\titlerunning{Adaptive Homing is in P}
\def\authorrunning{N. Kushik \& N. Yevtushenko}
\begin{document}
\maketitle

\begin{abstract}
Homing preset and adaptive experiments with Finite State Machines (FSMs) are widely used when a non-initialized discrete event system is given for testing and thus, has to be set to the known state at the first step. The length of a shortest homing sequence is known to be exponential with respect to the number of states for a complete observable nondeterministic FSM while the problem of checking the existence of such sequence (Homing problem) is PSPACE-complete. In order to decrease the complexity of related problems, one can consider adaptive experiments when a next input to be applied to a system under experiment depends on the output responses to the previous inputs. In this paper, we study the problem of the existence of an adaptive homing experiment for complete observable nondeterministic machines. We show that if such experiment exists then it can be constructed with the use of a polynomial-time algorithm with respect to the number of FSM states.
\end{abstract}

\section{Introduction}
Finite State Machines (FSMs) are widely used when deriving high quality tests for reactive discrete event systems. If a system is non-initialized then homing and synchronizing experiments with FSMs are used in order to set the system into the known state \cite{Sandberg}. Homing experiments can be preset and adaptive when the next input significantly depends on output responses produced to the previously applying inputs. As the underlying model for synchronizing experiments is a finite automaton without outputs, only preset synchronizing experiments are considered in various papers (see, for example \cite{Volkov}, \cite{Cerny}, \cite{Klyachko}). 

Homing experiments are well studied for deterministic FSMs where minimal length homing sequences are derived based on a truncated successor tree \cite{Gill}, \cite{Kohavi}. Any deterministic complete reduced FSM with $n$ states has a homing sequence of length up to $n(n - 1)/2$. A related detailed survey on deriving synchronizing sequences is given by Sandberg in \cite{Sandberg}. As usual, when performing 'gedanken' experiments with state models, a transition behaviour of the machine under experiment is supposed to be known \cite{Moore}.

Nowadays, nondeterministic FSMs are thoroughly studied for deriving tests with the guaranteed fault coverage. The reason is that the specification FSM can be nondeterministic according to a number of reasons. For example, it can be a corollary to the optionality, as it happens when extracting an FSM from corresponding RFC specifications \cite{IRC}, or non-determinism can occur according to the limited controllability and/or observability when testing a component of a modular system \cite{TestCom2006}. Correspondingly, there are a number of publications about test derivation against nondeterministic FSMs. An FSM is nondeterministic if at some state the machine has several transitions under a given input. If an FSM under test is non-initialized a homing preset or adaptive sequence has to be applied before a test sequence in order to set the FSM into the known state.

For preset homing experiments for nondeterministic FSMs, Kushik et al. \cite{CIAA2011} show that differently from deterministic FSMs a homing sequence does not necessarily exist for a complete reduced nondeterministic FSM and proposed an algorithm for deriving a preset homing sequence for a given observable nondeterministic FSM when such sequence exists. A tight upper bound on a shortest preset homing sequence is exponential with respect to to the number of FSM states while the problem of checking the existence of a preset homing sequence for nondeterministic FSMs (Homing problem) is PSPACE-complete \cite{Programmirovanie2014}.

In order to decrease the complexity of related problems, adaptive homing experiments can be used when a next input that is applied to a machine under experiment is selected based on the output responses to previously applied inputs. However, Hibbard \cite{Hibbard} showed that, in general, for deterministic machines adaptive homing experiments do not shorten the length of an applied input sequence. In other words, deterministic complete machines require adaptive homing experiments with the height of the same order as for the preset case. However, it is not the case for nondeterministic FSMs. For nondeterministic FSMs, the upper bound on the height of adaptive homing experiments and the complexity of checking the existence of a homing adaptive experiment can be reduced with respect to the preset case. In this paper, we study the problem of checking the existence of an adaptive homing experiment for a nondeterministic FSM and refer to this problem as an ``Adaptive Homing'' problem.

Similar to Agrawal et al. \cite{Agrawal} who first showed the existence of an unconditional polynomial-time algorithm for checking if an integer is prime or composite, we present an approach for checking and deriving (if possible) an adaptive homing test case that represents an adaptive homing experiment for a given nondeterministic FSM. The unconditional polynomial-time algorithm is based on efficient checking of the existence of a homing test case for each pair of the FSM states. The complexity of this algorithm significantly depends on the number of FSM inputs and outputs as well as on the number of FSM states. Assuming that the number of FSM inputs as well as the number of its outputs is polynomial with respect to the number of states, we prove the polynomial complexity of an Adaptive Homing problem.

The structure of the paper is as follows. Section 2 contains Preliminaries. In Section 3, an approach for checking the existence of a homing test case for a given pair of FSM states is presented and the complexity of the Adaptive Homing problem is evaluated. Section 4 concludes the paper.

\section {Preliminaries}
A \textit{(non-initialized) Finite State Machine (FSM)} \textbf{S} is a 5-tuple $(S, I, O, h_S)$, where $S$ is a finite set of states; $I$ and $O$ are finite non-empty disjoint sets of inputs and outputs; $h_S \subseteq S \times I \times O \times S$ is a $transition$ $relation$, where a 4-tuple $(s, i, o, s') \in h_S$ is a $transition$. 

An FSM \textbf{S} = $(S, I, O, h_S)$ is \textit{complete} if for each pair $(s, i) \in S \times I$ there exists a pair $(o, s') \in O \times S$ such that $(s, i, o, s') \in h_S$; otherwise, the machine is \textit{partial}. Given a partial FSM \textbf{S}, an input $i$ is a \textit{defined} input at state $s$ if there exists a pair $(o, s') \in O \times S$ such that $(s, i, o, s') \in h_S$. FSM \textbf{S} is \textit{nondeterministic} if for some pair $(s, i) \in S \times I$, there exist at least two transitions $(s, i, o_1, s_1)$, $(s, i, o_2, s_2) \in h_S$, such that $o_1 \ne o_2$ or $s_1 \ne s_2$. FSM \textbf{S} is \textit{observable} if for each two transitions $(s, i, o, s_1)$, $(s, i, o, s_2) \in h_S$ it holds that $s_1 = s_2$. FSM \textbf{S} is \textit{single-input} if at each state there is at most one defined input at the state, i.e., for each two transitions $(s, i_1, o_1, s_1)$, $(s, i_2, o_2, s_2) \in h_S$ it holds that $i_1 = i_2$, and FSM \textbf{S} is \textit{output-complete} if for each pair $(s, i) \in S \times I$ such that the input $i$ is defined at state $s$, there exists a transition from $s$ with $i$ for every output in $O$ \cite{Petrenko2011}. The FSM with the designated initial state $s_0$ is an \textit{initialized} FSM, written $(S, s_0, I, O, h_S)$. Given initialized FSMs \textbf{S} = $(S, s_0, I, O, h_S)$ and \textbf{P} = $(P, p_0, I, O, h_P)$, the FSM \textbf{P} is a submachine of \textbf{S} if $P \subseteq S$, $p_0 = s_0$ and $h_P \subseteq h_S$. A \textit{trace} of \textbf{S} at state $s$ is a sequence of input/output pairs of sequential transitions starting from state $s$. Given a trace $(i_1, o_1) \dots (i_k, o_k)$ at state $s$, the input projection $i_1 \dots i_k$ of the trace is a defined input sequence at state $s$. An initialized FSM \textbf{S} = $(S, s_0, I, O, h_S)$ is \textit{acyclic} if the set $Tr$(\textbf{S}/$s_0$) of traces at the intial state is finite, i.e., the FSM transition diagram has no cycles. As usual, for state $s$ and a trace $\gamma$, the $\gamma$-successor of state $s$ is the set of all states that are reached from $s$ by $\gamma$. If $\gamma$ is not a trace at state $s$ then the $\gamma$-successor of state $s$ is empty or we simply say that the $\gamma$-successor of state $s$ does not exist. For an observable FSM \textbf{S}, the cardinality of the $\gamma$-successor of state $s$ is at most one for any trace $\gamma$. Given a nonempty subset $S'$ of states of the FSM \textbf{S} and a trace $\gamma$, the $\gamma$-successor of the set $S'$ is the union of $\gamma$-successors over all $s \in S'$.

As in this paper we consider homing experiments with nondeterministic FSMs, in order to identify a state of a given non-initialized FSM after the experiment, a finite input sequence is applied to the FSM where the next input (except of the first one) of the sequence is determined based on the output of the FSM produced to the previous input. Formally, such an experiment can be described using a single-input output-complete FSM with an acyclic transition graph and similar to \cite{Petrenko2011} we refer to such an FSM as a \textit{test case}. A test case \textbf{P} is \textit{homing} for the set $S'$ of states of the FSM \textbf{S} if for each trace $\gamma$ from the initial state to a deadlock state of \textbf{P}, there exists a state $s$ of \textbf{S} such that for each state $s' \in S'$, the $\gamma$-successor of state $s'$ does not exist or the $\gamma$-successor of state $s'$ is $s$. If there exists a homing test case for the set $S'$ then the set $S'$ is \textit{adaptively homing} or simply a \textit{homing set}. If there exists a homing test case for the set $S$ of all states then FSM \textbf{S} is \textit{adaptively homing}.

Given an input alphabet $I$ and an output alphabet $O$, a \textit{test case} $TC(I, O)$ is an initially connected single-input output-complete observable initialized FSM \textbf{P} = $(P, I, O, h_P, p_0)$ with the acyclic transition graph. By definition, if $|I| > 1$ then a test case is a partial FSM. A test case $TC(I, O)$ over alphabets $I$ and $O$ defines an adaptive experiment with any complete FSM \textbf{S} over the same alphabets. A test case $TC(I, O)$ over alphabets $I$ and $O$ is a \textit{homing test case} for FSM \textbf{S} = $(S, I, O, h_S)$, if for each trace $\gamma$ of the test case from the initial state to a deadlock state the $\gamma$-successor of the set $S$ is a singleton.

\section {Adaptive Homing problem for a pair of FSM states}
An FSM is \textit{adaptively homing} \cite{STTT2014} if there exists an adaptive homing experiment. In general, given a test case \textbf{P}, the \textit{length} of the test case \textbf{P} is determined as the length of a longest trace from the initial state to a deadlock state of \textbf{P} and it specifies the length of the longest input sequence that can be applied to an FSM \textbf{S} during the experiment that is also often called the \textit{height} of the adaptive experiment. In this section, we discuss how an Adaptive Homing problem can be solved for a given pair of FSM states. This problem is stated as follows. Given a complete nondeterministic FSM \textbf{S} = $(S, I, O, h_S)$, and a pair $(s_i, s_j)$, the question is whether there exists a homing test case for a pair $(s_i, s_j)$. As the reply, there can be a homing test case for a pair $(s_i, s_j)$ or a message ``the pair $(s_i, s_j)$ is not adaptively homing''.

In this section, we focus on solving Adaptive Homing problem for a pair of FSM states, since if the FSM under experiment is observable then Adaptive Homing problem for FSM \textbf{S} can be reduced to this problem for each state pair. Given a state pair $(s_i, s_j)$, we propose a procedure for checking whether this pair of states is adaptively homing based on the corresponding FSM intersection. The complexity evaluation seems to be more straightforward when using the corresponding intersection than for the procedure proposed in \cite{STTT2014}. Without loss of generality, consider a pair $(s_1, s_2)$ of FSM states.

Given a complete observable FSM \textbf{S} = $(S, I, O, h_S)$ and two different states $s_1$ and $s_2$ of \textbf{S}, we derive the intersection \textbf{S}/$s_1$ $\cap$ \textbf{S}/$s_2$ = $(Q, (s_1,s_2), I, O, h_{\textbf{S}/s_1 \cap \textbf{S}/s_2})$ in a usual way. States of \textbf{S}/$s_1$ $\cap$ \textbf{S}/$s_2$ are pairs $(s_j, s_k)$, $j < k$, $j, k = 1, \dots, n$, and there is a transition $((s_j, s_k), i, o, (s'_j, s'_k))$ if and only if $s'_j \ne s'_k$ and $s'_j, s'_k$ are $io$-successors of states $s_j$ and $s_k$. If for all $o$ there are no such $io$-successors then a transition at the state $(s_j, s_k)$ under input $i$ is not defined.

\begin{proposition}
Given two states $s_1$ and $s_2$ of a complete observable FSM \textbf{S} = $(S, I, O, h_S)$ and the intersection \textbf{S}/$s_1$ $\cap$ \textbf{S}/$s_2$, states $s_1$ and $s_2$ are not adaptively homing if and only if the intersection \textbf{S}/$s_1$ $\cap$ \textbf{S}/$s_2$ has a complete submachine. 
\end{proposition}
\begin{proof}
$\Longleftarrow$ 
Let there exist a complete submachine \textbf{F} of the intersection \textbf{S}/$s_1$ $\cap$ \textbf{S}/$s_2$. By definition, this means that for every input sequence $\alpha$ there exists an output sequence $\beta$ such that the trace $\alpha/\beta$ takes the pair $(s_1, s_2)$ to another pair $(s_j, s_k)$, $s_j \ne s_k$, i.e., the pair $(s_1, s_2)$ is not adaptively homing.

$\Longrightarrow$ 
Consider states $s_1$ and $s_2$ which are not adaptively homing. Derive a subset $Q'$ of states of FSM \textbf{Q} = \textbf{S}/$s_1$ $\cap$ \textbf{S}/$s_2$ which are not adaptively homing. For each state $q \in Q'$ and each input $i \in I$, there exists a transition $(q, i, o, q')$ in \textbf{Q} for some $o \in O$ and $q' \in Q$ and moreover, at least one of such states $q'$ is in the set $Q'$; otherwise, state $q$ is adaptively homing. Since the initial state of is not adaptively homing, $(s_1, s_2) \in Q'$; and thus, the FSM \textbf{Q} has a complete submachine with the set $Q'$ of states. 
\end{proof}

\begin{corollary}
States $s_1$ and $s_2$ are adaptively homing if and only if the intersection \textbf{S}/$s_1$ $\cap$ \textbf{S}/$s_2$ has no complete submachine, i.e., each submachine has an input undefined in some state. 
\end{corollary}

\begin{proposition}
The Adaptive Homing problem for given two states $s_1$ and $s_2$ of a complete observable FSM \textbf{S} = $(S, I, O, h_S)$ is in P, when the number of FSM inputs/outputs is polynomial w.r.t. the number of FSM states. 
\end{proposition}
\begin{proof}
In order to estimate the complexity of the Adaptive Homing problem for two states $s_1$ and $s_2$, we evaluate the corresponding complexity as a function of the number $n = |S|$ of FSM states.
	The existence of a complete submachine can be checked by iterative removal from the intersection \textbf{S}/$s_1$ $\cap$ \textbf{S}/$s_2$ each state that has an undefined input along with its incoming transitions. If at the end, the initial state is also removed then the two given states are adaptively homing, otherwise they are not adaptively homing. The procedure is polynomial with respect to the number of states \cite{the_book} when the number of inputs and outputs are polynomial with respect to the number of states, and thus, the complexity of checking the existence of a complete submachine of \textbf{S}/$s_1$ $\cap$ \textbf{S}/$s_2$ is polynomial, since this machine has at most $n(n - 1)/2$ states.
\end{proof}

We mention, that the procedure of checking whether the pair of states is homing, given in Proposition 2, is similar to the one for checking the existence of an adaptive $(s_1, s_2)$-distinguishing strategy considered in \cite{Petrenko2011}. Once the existence of a homing test case is proven one can derive such test case in various ways. A reader may turn to \cite{Petrenko2011} where an algorithm for deriving an adaptive distinguishing test case is proposed or to \cite{STTT2014} where a general procedure for deriving a homing test case for a weakly initialized FSM is proposed. In both cases, the number of states of a shortest test case for a pair of FSM states is at most $n(n - 1)/2 + 1$, where the integer 1 is added for the designated deadlock state. As a test case is a single-input output complete FSM, the maximal number of the test case transitions equals the product of $(n (n - 1)/2)$ and $|O|$. In other words, the following proposition holds.

\begin{proposition}
Given a complete observable nondeterministic FSM \textbf{S} and states $s_1$ and $s_2$ of FSM \textbf{S}, if states $s_1$ and $s_2$ are adaptively homing then the number of transitions of a shortest homing test case does not exceed $(n (n - 1)/2) \cdot |O|$.
\end{proposition}

\begin{corollary}
Given a complete observable adaptively homing nondeterministic machine \textbf{S}, a shortest homing test case requires a polynomial size of memory for its storage, if the number of FSM outputs is polynomial with respect to the number of its states.
\end{corollary}

In \cite{STTT2014}, the following statement is established.

\begin{proposition}
\cite{STTT2014} A complete observable FSM \textbf{S} is adaptively homing if and only if each pair of two different states is homing.
\end{proposition}

Since for an FSM with $n$ states, the number of pairs of different states is $n (n -1)/2$, the above proposition immediately implies the following statement.

\begin{proposition}
The problem of checking of the existence of a homing test case for a complete observable FSM \textbf{S} is in P.
\end{proposition}

\textbf{Example}. Consider an FSM with a flow table in Table 1. By direct inspection, one can assure that each pair of states is adaptively homing. A flow table for the intersection \textbf{S}/$s_1$ $\cap$ \textbf{S}/$s_2$ that has no complete submachine, is shown in Table 2. States 1 and 3 can be homed by the input $i_2$ while States 2 and 3 can be homed by the input $i_1$.

By direct inspection, one can assure there does not exist a complete deterministic submachine for the FSM \textbf{S}/$s_1$ $\cap$ \textbf{S}/$s_2$ with a flow table in Table 2. Therefore, the FSM \textbf{S} with a flow table in Table 1 is adaptively homing. One of homing test cases for this machine is shown in Table 3.

\begin{table}	
\caption{FSM \textbf{S}}		
\begin{tabular}{|c|c|c|c|}	
\hline	
Input/State & 1 & 2 & 3 \\ \hline	
$i_1$ & $1/o_1$, $3/o_2$ & $2/o_2$, $3/o_1$ & $2/o_2$ \\ \hline
$i_2$ & $1/o_1$, $2/o_2$ & $3/o_1, o_2$ & $1/o_1$ \\ \hline
$i_3$ & $1/o_1$, $3/o_2$ & $2/o_1$ & $1/o_2$ \\ \hline	
\end{tabular}	
\end{table}

\begin{table}	
\caption{FSM \textbf{S}/$s_1$ $\cap$ \textbf{S}/$s_2$}		
\begin{tabular}{|c|c|c|c|}	
\hline	
Input/State & $\overline{1,2}$ & $\overline{2,3}$ & $\overline{1,3}$ \\ \hline	
$i_1$ & $\overline{1,3}/o_1$, $\overline{2,3}/o_2$ &   & $\overline{2,3}/o_2$ \\ \hline
$i_2$ & $\overline{1,3}/o_1$, $\overline{2,3}/o_2$ & $\overline{1,3}/o_1$ &   \\ \hline
$i_3$ & $\overline{1,2}/o_1$ &     & $\overline{1,3}/o_2$ \\ \hline	
\end{tabular}	
\end{table}

\begin{table}	
\caption{The homing test case for the FSM \textbf{S}}		
\begin{tabular}{|c|c|c|c|c|}	
\hline	
Input/State & $\overline{1,2,3}$ & $\overline{1,3}$ & $\overline{2,3}$ & $p$ \\ \hline	
$i_1$ & $\overline{1,3}/o_1$, $\overline{2,3}/o_2$ &   &   & \\ \hline
$i_2$ &    & $p/o_1, o_2$ & $p/o_1, o_2$   &  \\ \hline
\end{tabular}	
\end{table}

\section{Conclusion}
In this paper, we have shown that the problem of checking the existence of a homing test case for a given complete observable FSM is in P. Moreover, we have shown that for a pair of states of an adaptively homing FSM a shortest homing test case requires a polynomial size of memory for its storage, if the number of FSM outputs is polynomial with respect to the number of its states.

Keeping in mind that the procedure for deriving a homing test case for the set of all states of a complete observable FSM asks in fact for concatenatian of test cases for pairs of states, we can presume that for an adaptively homing FSM, a shortest homing test case requires a polynomial size of memory for its storage when the number of FSM outputs is polynomial with respect to the number of its states. Nevertheless, we leave the proof of this statement for the future work as well as the study whether the problem of deriving a shortest homing test case is also in P. Another interesting problem can be the exact characterization of all homing test cases in the form of an FSM, called a canonical homing test case, similar to the canonical separator for distinguishing test cases in \cite{Petrenko2011}.

\nocite{*}
\bibliographystyle{eptcs}
\bibliography{generic_mine}
\end{document}